\documentclass[letterpaper, 10 pt, conference]{ieeeconf}  %

\IEEEoverridecommandlockouts                              %
\usepackage{graphics} %
\usepackage{epsfig} %
\usepackage{mathptmx} %
\usepackage{times} %
\usepackage{amsmath} %
\usepackage{amssymb}  %

\usepackage{cite}
\usepackage{amsmath,amssymb,amsfonts}
\usepackage{bm}
\usepackage{algorithmic}
\usepackage{graphicx}
\usepackage{textcomp}
\usepackage{xcolor}
\usepackage{amsthm}
\usepackage[ruled,vlined]{algorithm2e}
\usepackage{comment}
\usepackage{caption}
\usepackage{subcaption}
\usepackage{array}
\usepackage{multirow}
\usepackage{lipsum}
\usepackage{hhline}
\usepackage{adjustbox}
\usepackage{cuted}
\let\labelindent\relax
\usepackage{enumitem}

\newtheorem{theorem}{Theorem}[section]
\newtheorem{corollary}{Corollary}[theorem]
\newtheorem{lemma}[theorem]{Lemma}
\newtheorem{assumption}{Assumption}[section]
\newtheorem{remark}{Remark}[section]

\def\BibTeX{{\rm B\kern-.05em{\sc i\kern-.025em b}\kern-.08em
    T\kern-.1667em\lower.7ex\hbox{E}\kern-.125emX}}

\title{\LARGE \bf
Belief-Driven Search in Connectivity-Constrained Mobile Sensor Networks with Convergence Guarantees
}

\author{George J. Ziavras, Ziqiao Zhang, Shreyas Sundaram and Christopher G. Brinton%
\thanks{This work has been submitted to the IEEE for possible publication. Copyright may be transferred without notice, after which this version may no longer be accessible.}%
\thanks{G. Ziavras, Z. Zhang, S. Sundaram, and C. Brinton are with
the Elmore Family School of Electrical and Computer Engineering, Purdue University, West
Lafayette, IN 47906, USA. Email: \{gziavras, zhan5780, sundara2, cgb\}@purdue.edu.}%
\thanks{This work was supported by Saab Inc. and Saab AB.}%
}%

\begin{document}

\maketitle
\thispagestyle{empty}
\pagestyle{empty}

\begin{abstract}
Multi-agent search and coverage is critical to autonomous exploration, search-and-rescue, and defense applications. This paper investigates settings where a mobile sensor network is tasked with coordinating to discover the position of a stationary target within a bounded region. We develop a Bayesian update scheme where agents iteratively combine local uncertainties derived from their sensor measurements through inverse variance weighting to get a better global understanding of where the target may be located given some prior belief. Network agents' trajectories are governed by a primal-dual gradient update rule that simultaneously minimizes the expected sensing uncertainty and enforces network connectivity among the agents. We prove that the belief converges to the target under mild planner conditions, with an explicit exponential rate of false-hypothesis rejection and a bound on expected search time. Simulation results show that we can successfully find the target and collapse belief to the ground truth target location while ensuring connectivity of our sensor network. They also empirically confirm our theoretical bounds and planner conditions. Finally, we compare against a naive lawnmower sweeping strategy, showing that our method locates the target more quickly when the prior is informative.  
\end{abstract}

\section{INTRODUCTION}

Autonomous multi-agent systems are increasingly vital for operations in high-stakes, uncertain environments, ranging through autonomous exploration, environmental monitoring, search-and-rescue, and defense applications. A common objective in these settings is to locate a target whose position within some region is uncertain. We build on related coverage control to establish an active search algorithm that performs this task. Coverage control often seeks to distribute stationary agents based on a fixed importance map to optimize some objective. In contrast, active search benefits from agents actively learning, adjusting their environment beliefs based on sensor feedback, to learn a target's location (e.g. Fig.~\ref{fig:overview}). 
\begin{figure}[t]
    \centering
    \includegraphics[width = 0.48\textwidth]{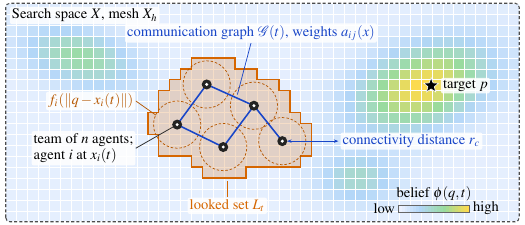}
    \caption{System model. A connected team of $n$ agents searches the mesh $X_h$ for a target at an unknown cell $p$. Fusing their measurements (dashed: one agent's range), the team senses $L_t$, whose reports $y_t(q)$ update the shared belief $\phi(q,t)$ by~\eqref{eq:phi}.}
    \label{fig:overview}
\end{figure}
 
\noindent\textbf{Key Challenges.} Multi-robot search under communication constraints poses two core difficulties:

\textit{1. Information gain versus network connectivity.} To maximize the rate of information gain, uncertainty minimization incentivizes agents to spread out and cover unvisited regions. Fusing measurements and maintaining a shared belief, however, requires a connected communication network. %
\textit{2. The need for tractable, differentiable constraints.} We seek to incorporate connectivity constraints as a functional requirement for information exchange. We express connectivity as differentiable constraints, so that a primal-dual gradient method can update agent trajectories continuously while preserving connectivity. %

\noindent\textbf{Main Contributions.} Our contributions are as follows.
\begin{itemize}[leftmargin=*]
\item \textbf{Inverse-Variance Sensor Fusion:} We develop a sensing model that replaces Voronoi partitioning with a smoother global measurement uncertainty, enabling more flexible gradient-based control. This uncertainty also defines detection probabilities given some fixed false-alarm rate, driving an exact Bayesian update of the belief.
\item \textbf{Connectivity-Constrained Planning:} We encode network connectivity as a differentiable constraint, solving the resulting planning problem with a primal-dual gradient method, showing empirically in simulation that a team starting connected remains connected. 
\item \textbf{Convergence Analysis:} We prove that, for any planner that does not anticipate future measurements, the probability of ever confirming a wrong location is explicitly bounded. We introduce \emph{Aggregated Global Identifiability} (AGI), a weaker revisiting condition than full coverage (cells with non-negligible belief, or the target cell, must be sensed within each finite window). Under this condition, the belief concentrates on the target. Belief in false locations decays exponentially at a rate determined by the sensor's false-alarm probabilities and missed-detection threshold, and we bound the expected search time. 
\end{itemize}

\section{RELATED WORKS}
Distributed search \& coverage control has traditionally used Voronoi partitions to tessellate the search space \cite{agerman2025primal,schwager2009decentralized,pratissoli2025distributed, guruprasad2011automated}. Density-Driven Optimal Control (D2OC) instead applies optimal transport to match agent distributions to non-uniform priority maps \cite{SeoL2025}, but assumes these maps are static or prescribed. AI-augmented extensions update these maps under uncertain priors \cite{lee2026ai} at the cost of added computation and training. Effective coverage control \cite{hussein2007effective} guarantees that every point is sensed above a capacity threshold which controls how much each point is sensed over time, using a strict maximum sensing radius and a measurement-driven penalty over a fixed underlying belief distribution. None of these models the cumulative information gain from sensor overlap and how that gain alters the localization belief. 

We fuse measurements from connected sensors into a global uncertainty estimate and use Bayesian updates to model how sensing changes the belief. Close to this, \cite{prabhakar2025decoupled} applies inverse-covariance fusion for decentralized state estimation. We trade fully decentralized covariance fusion in favor of a simpler differentiable objective built on a distance-coupled local uncertainty function. This physics-informed model sums inverse-variances across overlapping sensors, as is standard for independent Gaussian measurements \cite{battistelli2014kullback}, an assumption we make for our sensing model. Weighted variants of this method can be used to accommodate heterogeneous sensors, as in \cite{Qiu2024multisensor}. Bayesian belief updates also drive information-theoretic search \cite{julian2012distributed,hu2013vision}. Unlike these, we model sensor-overlap gain in closed form and enforce multi-agent connectivity.

Online sampling-based informative path planners like IA-TIGRIS build candidate trajectory trees through informed sampling and select trajectories by scoring metrics like information gain \cite{moon2025iatigris}. They offer strong global exploration in non-convex environments and resist local minima, but scale poorly to joint multi-agent, constraint-aware state spaces. Our gradient-based planner avoids this but can stall at saddle points, and in flat or low-belief regions. We mitigate this with perturbed gradients
and an adaptively varied gradient pull region (Algorithm \ref{alg:ppala}).

\section{PROBLEM FORMULATION}

\subsection{Search Space and Sensing Model}
We consider $n$ sensors (agents) in a convex, compact search space $X \subset \mathbb{R}^d$, and its discretized mesh $X_h$ of spacing $h$ with $c = |X_h|$ cells. A single stationary target occupies an unknown cell $p\in X_h$. Time is discrete, $t\in\{0,1,2,...\}$, and $x(t) = [x_1(t)^\intercal, x_2(t)^\intercal,...,x_n(t)^\intercal]^\intercal\in ({X})^n$ collects the sensor positions, with $x_i(t) \in X$ the position of sensor $i$. Sensing uncertainty is modeled by a continuously differentiable, non-decreasing function $f_{i}: \mathbb{R}_+ \to \mathbb{R}_+$, where $f_{i}(||q-x_i(t)||)$ is the variance of sensor $i$'s measurement at cell $q$ and $\|\cdot\|$ is the Euclidean norm. Individual measurement noise is independent across sensors, cells, and time.

The belief $\phi(q,t)$ is a probability mass function on $X_h$ giving the probability that the target is at $q$ at time $t$, starting from a prior $\phi(\cdot, 0)$. Our goal is to update $\phi$ and the trajectories $x_i(t)$ jointly so that $\phi(\cdot, t)$ converges to the point mass $\phi^*(q) = \mathbf{1}\{q = p\}$ at the target.                       

\subsection{Global Uncertainty and Search Objective}
\subsubsection{Inverse Variance Weighting}
Unlike coverage control that relies on Voronoi partitions to assign each cell $q$ to its nearest sensor, we fuse all sensors' measurements under $f(q, x(t))$. Specifically, we employ Inverse Variance Weighting \cite{xiu2024multisensor} (under our assumption of each agent being equipped with identical sensors) on the individual sensor variances to obtain the combined unbiased minimum-variance,
\begin{equation}
    \label{eq:fqx}
    f(q,x(t)) = \left(\sum_{i = 1}^n \frac{1}{f_{i}(||q - x_i(t)||)}\right)^{-1}.
\end{equation}
Since $f \leq \min_i f_i$, a group of agents measures every cell at least as well as its best member. The search objective is the expected fused uncertainty under the current belief,
\begin{equation}
\label{eq:H_cost}
    \mathcal{H}_t(x) = \sum_{q\in X_h}f(q,x)\phi(q,t).
\end{equation}

\subsubsection{Network Connectivity}
Our other key objective is to maintain network connectivity, for which we employ results in~\cite{zavlanos2007potential,agerman2025primal, siligardi2019robust}. Letting $L(\mathcal{G}(t)) = \Delta(x) - A(x)$ be the weighted Laplacian of the communication graph, where the weights $a_{ij}(x)$ of the adjacency matrix $A(x)$ are a sigmoid $a_{ij}(x)=\big(1+e^{-s\,(r_c-\|x_i-x_j\|)}\big)^{-1}$ with steepness $s$, which decays towards zero beyond the communication range $r_c$, and $\Delta(x)$ is the diagonal degree matrix. Let $P \in \mathbb{R}^{n \times (n-1)}$ have orthonormal columns orthogonal to $\mathbf{1}_n$. Projecting onto this subspace removes the Laplacian's zero eigenvalue, so $\det(P^TL(x)P) = \lambda_2\cdots\lambda_n$, which by the matrix-tree theorem equals $n$ times the weighted number of spanning trees. It is therefore positive if and only if the graph is connected, so we impose
\begin{equation}
    g_{conn}(x(t)) = \zeta - \det(P^T L(\mathcal{G}(t))P)
\end{equation}
where $\zeta > 0$ sets the required margin. The reason we utilize this constraint over other common methods such as maintaining $\lambda_2$ above some threshold is due to its smoothness and ability to be differentiable in all graph configurations. Furthermore, counting spanning trees measures overall connectivity, not just the weakest cut, which rewards redundant paths and maintains a stronger network. The flip side is that a single weakening link in a large group changes $g_{conn}$ only slightly, which the safeguards in Algorithm \ref{alg:ppala} address.  

\subsubsection{Collision Avoidance}
To avoid collisions, every pair of sensors keeps a minimum separation $d_c>0$
\begin{equation}
    g_{ij}(x(t)) = d_c - ||x_i(t) - x_j(t)|| \leq 0, \ \ i\neq j.
\end{equation}

\subsection{Measurement and Belief Update}
At every timestep the team takes a fused measurement, where each cell $q$ of a looked set $L_t\subseteq X_h$ returns a binary report $y_t(q)\in \{0, 1\}$ ($1 = $ detection). The looked set determines how a cell's belief will be updated by Bayes' rule~\eqref{eq:phi}. We define the set we use for $L_t$ below. 

We model the detector as a Neyman-Pearson test for a known signal in white Gaussian noise with false-alarm probability $P_{fa}$ and fused $SNR_{linear} = \gamma_s/f(q,x)$, where $\gamma_s > 0$ is a signal-strength scale. Under these assumptions, its miss and detection probabilities are respectively:
\begin{align}
    P_e(SNR_{linear}^{-1}) &= \mathbf{Q}(\sqrt{SNR_{linear}} - \mathbf{Q}^{-1}(P_{fa})), \label{eq:Pe}\\
    P_d(SNR_{linear}^{-1}) &= 1 - P_e(SNR_{linear}^{-1}),
\end{align}
 where $\mathbf{Q}(\cdot)$ is the Gaussian tail function. Since $\mathbf{Q}$ is decreasing, $P_d > P_{fa}$ and far from every sensor, $P_d \to P_{fa}$ modeling how such reports carry no information. A cell $q\in X_h$ is \emph{sensed} at $t$ if $q\in X_{h,t,active}=\{q:P_e(q,x(t))\le\delta\}\subseteq L_t$, i.e.\ it falls in the region where $P_d\ge1-\delta$. Our work sets $L_t = X_{h,t,active}$. The target cell $p$ reports $1$ with probability $P_d(p, x(t))$, and every other cell with probability $P_{fa}$, independently across cells. By Bayes' rule: 
\begin{equation}\label{eq:phi}
\begin{split}
  \phi(q,t+1) &= \frac{\phi(q,t)\,\ell_t(q)}{\sum_r\phi(r,t)\ell_t(r)},\quad \\
  \ell_t(q) &= \begin{cases}
    P_d(q,x(t))/P_{fa}, & q\in L_t,\ y_t(q)=1,\\ 
    P_e(q,x(t))/(1-P_{fa}), & q\in L_t,\ y_t(q)=0,\\ 
    1, & q\notin L_t.
  \end{cases}
\end{split}
\end{equation}
A detection raises a cell's belief, a miss lowers it, and unlooked cells change only through normalization. Updates of this form appear in \cite{mitra2020new, julian2012distributed, hu2013vision, castanon1995optimal}.

\subsection{Optimization Problem}
At each time $t$, the search problem is
\begin{equation}\label{eq:opt}
\begin{aligned}
    \min_{x \in (X)^n} \quad & \mathcal{H}_t(x) \\
    \text{s.t.} \quad & g_{conn}(x) \leq 0 ,
    & g_{ij}(x) \leq 0, \quad \forall \, i\neq j .
\end{aligned}
\end{equation}
We turn to solving this problem in the next section.

\section{ALGORITHM METHODOLOGY}
\label{sec:alg}
In this section, we detail our optimization algorithm used to solve~\eqref{eq:opt}. To handle the non-convex nature of the objective function $\mathcal{H}(x(t))$ and the complex connectivity constraint $g_{conn}(x(t))$, we adopt the Primal-Dual Proximal Augmented Lagrangian Algorithm (PPALA) developed in \cite{Kim2024FastSingleLoop,agerman2025primal}. We drop the index $t$ when the context is clear. Unlike the static setting of~\cite{agerman2025primal}, where PPALA converges to a KKT point of fixed sensor positions, our objective $\mathcal{H}_t$ changes at every step as the belief is updated, so we do not claim convergence of the positions. Instead, Section~\ref{sec:conv} shows that the belief concentrates on the target for any planner that satisfies \textbf{AGI}.

\subsection{The Augmented Lagrangian}
The augmented Lagrangian $\mathcal{L}_\rho(x,u_G, \lambda_G)$ takes the form
\begin{equation} \label{eq:augLag}
    \mathcal{L}_\rho(x, u_G, \lambda_G) = J(x) +  w\big(\lambda_G(g_G(x) + u_G) + \frac{\rho}{2}(g_G(x) + u_G)^2\big),
\end{equation}
where we have $g_G(x) = g_{conn}(x)$ (our connectivity constraint), a slack $u_G\geq 0$ turning $g_G\leq 0$ into the equality $g_G + u_G = 0$, $\lambda_G \geq 0$ is the Lagrange multiplier, $\rho > 0$ is the penalty parameter, and $w>0$ is a weight balancing search against connectivity. Here $J(x) = \sum_{q\in R_G}f(q,x)\phi(q,t)$ is $H_t$ restricted to the pull region $R_G\subseteq X_h$ detailed in Algorithm \ref{alg:ppala}. Collisions are enforced directly, by projecting onto the minimum separation $d_c$ after each step. PPALA regularizes the dual variable toward a proximal center $\mu_G$. The multiplier is $\lambda_G=\arg\max_{\lambda}\{\lambda(g_G+u_G)-\tfrac{1}{2\rho}(\lambda-\mu_G)^2\}=\mu_G+\rho(g_G+u_G)$, and $\mu_G$ tracks $\lambda_G$ as a moving average with rate $\sigma_k\in(0,1)$. Anchoring $\lambda_G$ to the slowly moving $\mu_G$ damps oscillation between search and connectivity \cite{Kim2024FastSingleLoop}.
Algorithm \ref{alg:ppala} lists the full procedure composed of a shared belief update, followed by a normalized primal step with noise and a connectivity safeguard, and a dual update.

\begin{algorithm*}[!t]
    \small
    \caption{Sensor Movement and Belief Update Algorithm based on PPALA}
    \label{alg:ppala}
    \begin{algorithmic}[1]
    \STATE \textbf{Input:} Initial positions $x(0)$ forming one connected group, prior $\phi(\cdot,0)$, and $\lambda_G=\mu_G=u_G=0$.
    Sensing: $P_{fa}$, $\gamma_s$, $\delta$. Connectivity: $r_c$, $s$, $\zeta$. Pull region: $r_s$, $\gamma$, $\varepsilon_g$. Motion: $d_{\min}$, $d_{\max}$, $d_c$, $\sigma$. PPALA: $w$, $\rho$, $\tau$, $\sigma_k$. Horizon: $T_{\max}$.
    \FOR{$t = 0, 1, \dots, T_{\max}$}
        \STATE \textbf{Groups:} Agents within range $r_c$ form connected groups $G$. Connected agents share $x_i(t)$, their beliefs and $(\lambda_G,u_G,\mu_G)$. 
        \STATE \textbf{Sensing and belief:} $G$ senses $L_t = X_{h,t,active}=\{q\in X_h:P_e(q,x(t))\le\delta\}$.
        Each $q$ returns $y_t(q)\in\{0,1\}$ and $\phi$ is updated by~\eqref{eq:phi}.
        \STATE \textbf{Pull region:} $R_G=X_{h,t,active}$ if it contains $q$ s.t. $\phi(q,t+1)>\varepsilon_g\max_{q'}\phi(q',t+1)$. Otherwise $R_G$ adds all cells within distance $\gamma^kr_s$ of $G$, for the smallest $k\in\mathbb{N}$ that includes such a cell. ($R_G$ affects only the gradient, not the belief.)
        \STATE \textbf{Primal update:} With $J(x)=\sum_{q\in R_G}f(q,x)\,\phi(q,t+1)$ and $g_G(x)=\zeta-\det(P^\intercal L_G(x)P)$, where $L_G$ is the weighted Laplacian of $G$ and $P$ has orthonormal columns orthogonal to $\mathbf{1}$,
        \begin{gather*}
        \nabla_{x_i}\mathcal{L}_\rho = \nabla_{x_i}J + w\big(\lambda_G + \rho(g_G + u_G)\big)\nabla_{x_i}g_G,
        \qquad
        \eta_{x,i} = \mathrm{clip}\Big(d_{\max}\tfrac{\lVert\nabla_{x_i}\mathcal{L}_\rho\rVert}{\max_{j\in G}\lVert\nabla_{x_j}\mathcal{L}_\rho\rVert},\,d_{\min},\,d_{\max}\Big),\\
        x_i(t+1) = x_i(t) - \eta_{x,i}\,\frac{\hat{\nabla}_{x_i}\mathcal{L}_\rho + \xi_i}{\lVert\hat{\nabla}_{x_i}\mathcal{L}_\rho + \xi_i\rVert},
        \qquad \hat{\nabla}_{x_i}\mathcal{L}_\rho=\frac{\nabla_{x_i}\mathcal{L}_\rho}{\lVert\nabla_{x_i}\mathcal{L}_\rho\rVert},\quad \xi_i\sim\mathcal{N}(0,\sigma^2I),
        \end{gather*}
        then enforce separation $d_c$ and project onto $X$. \textbf{Safeguard:} if $x(t+1)$ would disconnect $G$, recompute it with $w\leftarrow2^mw$, $m=1,2,\dots$, until $G$ stays connected or $x(t+1)$ stops changing; then halve all steps in $G$ (up to 10 times); then hold position. Hence a group that starts connected stays connected.
        \STATE \textbf{Dual update:} With $g_G=g_G(x(t+1))$:\quad
        $u_G \leftarrow \max\{0,\,u_G - \tau(\lambda_G + \rho(g_G + u_G))\}$,\quad
        $\mu_G \leftarrow \mu_G + \sigma_k(\lambda_G - \mu_G)$,\quad
        $\lambda_G \leftarrow \mu_G + \rho(g_G + u_G)$.
    \ENDFOR
    \end{algorithmic}
\end{algorithm*}

\section{ANALYSIS}
\label{sec:conv}
The primal step of Algorithm~\ref{alg:ppala} moves each agent toward the belief in $R_G(t)$, weighted by its share of the fused sensing at each cell. On this nonconvex, time-varying objective, the descent can \emph{(i)} stall at critical points where the weighted belief balances about an agent, \emph{(ii)} lose direction once nearby belief is depleted, while \emph{(iii)} the connectivity metric reacts slowly to a single weakening link in a large group. Algorithm~\ref{alg:ppala} mitigates these with gradient noise, a pull region that widens when nearby belief is depleted, and a step-size safeguard with a connectivity radius below the true communication range. These measures make recurrent sensing plausible but not provable.
We therefore treat the planner as a black box, state our results under explicit assumptions on which cells it senses (Assumption~\ref{as:AGI}), and verify these assumptions empirically in Section~\ref{sec:sim}. 

\subsection{Gradient Analysis}
We adopt the quadratic sensing cost $f_{local}(\|q-x_i(t)\|)=\tfrac12\|q-x_i(t)\|^2+\epsilon$, standard in
coverage control \cite{schwager2006distributed,schwager2009decentralized}, where $\epsilon>0$ enforces
$f_{local}>0$. With $f(q,x(t))=\big[A_i+1/f_{local}(\|q-x_i(t)\|)\big]^{-1}$ and
$A_i=\sum_{j\neq i}1/f_{local}(\|q-x_j(t)\|)$, differentiating gives
\begin{equation}\label{eq:df_x}
  \nabla_{x_i(t)}f(q,x(t))=\frac{4\,(x_i(t)-q)}{\big(A_i(\|q-x_i(t)\|^2+2\epsilon)+2\big)^2}=w_q\,(x_i(t)-q),
\end{equation}
where $w_q=(f/f_{local})^2\in(0,1]$. $w_q$ is the square of the ratio describing the sensor fusion contribution compared to just agent $i$. It is close to $1$ when agent $i$ dominates sensing at $q$.

\subsection{Trajectory Behavior and Suboptimal Solutions}
During a step, we hold the set of the pull region fixed and evaluate on mesh $X_h$ of size $h$, with $\phi(\cdot,t)$ a probability mass function on $X_h$. By~\eqref{eq:df_x} the descent direction is
\begin{equation}\label{eq:df_x2}
  -\!\!\sum_{q\in R_G(t)}\!\!\phi(q,t)\,\nabla_{x_i(t)}f(q,x(t))=\sum_{q\in R_G(t)}\!\!w_q\,\phi(q,t)\,(q-x_i(t)).
\end{equation}
This is a vector sum from $x_i(t)$ toward each cell in $q$, weighted by their belief ($\phi$) and by agent $i$'s responsibility for them ($w_q$), where $R_G(t)$ is the pull region described in Algorithm \ref{alg:ppala}. For a lone agent ($w_q\equiv1$) it points at the centroid of belief on $R_G(t)$. %

Three suboptimal behaviors follow. \emph{(i) Critical points.} \eqref{eq:df_x2} vanishes when the weighted mass $\phi\,w_q$ balances about
$x_i(t)$, and saddles points are also possible with multiple agents. %
Noise may help here, although
perturbed-gradient guarantees \cite{jin2017escape} assume smoothness, which we cannot assume while $R_G(t)$ changes between steps. These call for an explicit escape mechanism. (ii) Dead regions.
When sensing depletes nearby belief, the gradient carries no information about belief outside $R_G(t)$. Switching to a gradient over all of $X_h$ would cause other stall issues, since
by~\eqref{eq:df_x2} a lone agent is drawn to the centroid of belief, which may be empty. In Algorithm \ref{alg:ppala}, we exploit the discontinuous property by varying radii for $R_G(t)$ when stuck, which showed empirical promise in treating \emph{(i)} and \emph{(ii)} in simulation. \emph{(iii) Slow connectivity response.}
The $\det(P^TLP)$ constraint reacts slowly to a single weakening link within large groups. A safeguard in step-size that rejects disconnecting moves, and a constraint radius below true communication range, avoid this in Algorithm \ref{alg:ppala}.
\subsection{Proof of Single-Target Convergence}
\paragraph{Model} The target is in a fixed, unknown cell $p\in X_h$, with $c=|X_h|$. Note from~\eqref{eq:phi} $\ell_t\in(0,\infty)$, so $\phi(\cdot,0) > 0$ implies $\phi(\cdot,t) > 0$ for all $t$. %
Let $\mathcal{F}_t$ denote the history up to step $t$, such as the reports $y_0,\dots,y_{t-1}$ together with any randomness used by the planner before $t$. For $q\neq p$, let
$R_t(q)=\phi(q,t)/\phi(p,t)$ and let $N_t(r)=|\{s<t:r\in X_{h,s,active}\}|$ count the sensed looks at $r$.

\begin{assumption}\label{as:std}
\begin{itemize}
  \item[(A1)] Probabilities $P_{fa}, P_d(q,x)$ satisfy $0<P_{fa}\le P_d(q,x)<1$. Given $\mathcal{F}_t$, the reports $\{y_t(q)\}_{q\in L_t}$ are independent with
        $y_t(p)\sim\mathrm{Bern}(P_d(p,x(t)))$ and $y_t(q)\sim\mathrm{Bern}(P_{fa})$ for $q\neq p$.
  \item[(A2)] State $x(t)$ and set $L_t$ are $\mathcal{F}_t$-measurable (determined from the past information in $\mathcal{F}_t$). All agents share belief~\eqref{eq:phi}, i.e., we assume information is pooled via connectivity. %
  \item[(A3)] Prior $\phi(q,0)$ satisfies $\phi(q,0)>0$ for all $q$, and $\delta<1-P_{fa}$.
\end{itemize}
\end{assumption}

Two cells $p_h,q_h\in X_h$ are \emph{distinguishable over $[t,t+T]$} if
$\{p_h,q_h\}\cap\bigcup_{\tau=t}^{t+T}X_{h,\tau,active}\neq\emptyset$, sensing either cell separates the two hypotheses. We relax global identifiability \cite{mitra2020new} to a windowed, belief-weighted form.

\begin{assumption}[Aggregated Global Identifiability]\label{as:AGI}
There exist $T<\infty$ and $\epsilon_{th}>0$ such that, almost surely, for every $t$ and every $q\neq p$
with $\phi(q,t)\ge\epsilon_{th}$, the cells $p$ and $q$ are distinguishable over $[t,t+T]$.
\end{assumption}

 \emph{Strong form}: The strong form of Assumption \ref{as:AGI} is that every $q$, or $p$, is sensed infinitely often (periodic coverage). This is only used in Theorem \ref{thm:conv}(i).

\paragraph{Constants} With $D(a\|b)=a\log\frac ab+(1-a)\log\frac{1-a}{1-b}$ (Bernoulli KL divergence) and
$h_\lambda(a,b)=a^{1-\lambda}b^\lambda+(1-a)^{1-\lambda}(1-b)^\lambda$ (Chernoff coefficient between Bernoulli distributions), define $\kappa_0=D(P_{fa}\|1-\delta)$, $\kappa_1=D(1-\delta\|P_{fa})$,
$\beta^\lambda_0=-\log h_\lambda(P_{fa},1-\delta)$, $\beta^\lambda_1=-\log h_\lambda(1-\delta,P_{fa})$ and
$\beta_\delta=\beta^{1/2}_0=\beta^{1/2}_1$ (Bhattacharyya distance). 

\begin{lemma}[Martingales]\label{lem:mart}
Under (A1)--(A3): (i) $R_{t+1}(q)=R_t(q)Z_t$ with $Z_t=\ell_t(q)/\ell_t(p)$. For $\lambda\in(0,1]$, $\mathbb{E}[Z_t^\lambda|\mathcal{F}_t]=h_\lambda(P_{fa},\pi_t(q))\,h_\lambda(\pi_t(p),P_{fa})$, where $\pi_t(r)=P_d(r,x(t))$ if
$r\in L_t$ and $P_{fa}$ otherwise. (ii) With $B^\lambda_t=-\sum_{s<t}\log\mathbb{E}[Z_s^\lambda|\mathcal{F}_s]$, where $\mathbb{E}$ denotes expectation, the process
$G^\lambda_t=R_t(q)^\lambda e^{B^\lambda_t}$ is a nonnegative martingale. For $\lambda=1$ this is $R_t(q)$ itself.
\end{lemma}
\begin{proof}
(i) Normalizations cancel in the ratio. By (A1), $\ell_t(q)$ and
$\ell_t(p)$ are independent given $\mathcal{F}_t$, with
$\mathbb{E}[\ell_t(q)^\lambda |\mathcal{F}_t]=P_{fa}(\pi_t(q)/P_{fa})^\lambda+(1-P_{fa})((1-\pi_t(q))/(1-P_{fa}))^\lambda=h_\lambda(P_{fa},\pi_t(q))$ and
$\mathbb{E}[\ell_t(p)^{-\lambda}|\mathcal{F}_t]=h_\lambda(\pi_t(p),P_{fa})$. An unlooked cell has factor $1=h_\lambda(P_{fa},P_{fa})$, and
$h_1\equiv1$. (ii) By \emph{(i)}, $\mathbb{E}[Z_t^\lambda|\mathcal{F}_t]$ depends only on $x(t)$ which is $\mathcal{F}_t$-measurable by (A2). Therefore, $B^\lambda_{t+1}$ is $\mathcal{F}_t$-measurable and $\mathbb{E}[G^\lambda_{t+1}|\mathcal{F}_t]=R_t^\lambda e^{B^\lambda_{t+1}}\mathbb{E}[Z_t^\lambda|\mathcal{F}_t]=G^\lambda_t > 0$ since $R_t > 0$, so $G_t^\lambda$ is a non-negative martingale.
\end{proof}

\begin{lemma}[Guaranteed Evidence per Look]\label{lem:info}
Under (A1)--(A3) and for $\lambda\in(0,1)$:
$B^\lambda_t\ge\beta^\lambda_0N_t(q)+\beta^\lambda_1N_t(p)$, $\beta^\lambda_i>0$, and $\beta^\lambda_i/\lambda\uparrow\kappa_i$ as $\lambda\downarrow0$.
\end{lemma}
\begin{proof}
By H\"older, $h_\lambda(a,b)\le1$, with equality iff $a=b$. So each step adds $\ge0$ to $B^\lambda_t$, and $\beta^\lambda_i>0$
since $1-\delta\neq P_{fa}$ by (A3). Since $\partial_bh_\lambda<0$ for $b>a$ and $\partial_ah_\lambda<0$ for $a>b$, at step $s$ a sensed $q$
($\pi_s(q)\ge1-\delta>P_{fa}$) gives $h_\lambda(P_{fa},\pi_s(q))\le h_\lambda(P_{fa},1-\delta)$, and a sensed $p$ gives
$h_\lambda(\pi_s(p),P_{fa})\le h_\lambda(1-\delta,P_{fa})$. Summing over $s<t$ gives the bound on $B_t^\lambda$. Finally, $\Lambda(\lambda)=\log h_\lambda(a,b)$ is the (convex) cumulant
generating function of the log-likelihood ratio of one report, with $\Lambda(0)=\Lambda(1)=0$. The chord slope of a convex function from the origin is nondecreasing, therefore
$-\Lambda(\lambda)/\lambda$ is nonincreasing in $\lambda$, with limit $-\Lambda'(0)=D(a\|b)$.
\end{proof}

\begin{theorem}[False confirmations]\label{thm:safe}
Under (A1) and (A3), for any planner satisfying (A2): (i) $\phi(\cdot,t)\to\phi_\infty$ almost surely (a.s.); (ii) for $\eta\in(0,1)$,
$\mathbb{P}\big(\exists t,\exists q\neq p:\phi(q,t)\ge1-\eta\big)\le\frac{\eta}{1-\eta}\cdot\frac{1-\phi(p,0)}{\phi(p,0)}$, where $\mathbb{P}$ denotes probability.
\end{theorem}
\begin{proof}
(i) By the martingale convergence theorem, each martingale $R_t(q)\ge0$ converges a.s. It can be seen that $\phi(p,t)=(1+\sum_{q\neq p}R_t(q))^{-1}$ and
$\phi(q,t)=R_t(q)\phi(p,t)$, therefore the beliefs converge as well. (ii) $\phi(q,t)\ge1-\eta$ forces $R_t(q)\ge(1-\eta)/\eta$. Ville's Inequality ensures $\mathbb{P}[\exists t:\phi(q,t) \geq 1-\eta ] \leq R_0(q)\eta/(1-\eta)$. Sum over $q$.
\end{proof}

\begin{theorem}[Convergence]\label{thm:conv}
Under (A1)--(A3): (i) on the event that $p$ is sensed infinitely often, or under the strong form, $\phi(p,t)\to1$ a.s.; (ii) under Assumption~\ref{as:AGI}, a.s.\ $\lim_t\phi(q,t)\le\epsilon_{th}$ for every $q\neq p$. Hence $\lim_t\phi(p,t)\ge1-(c-1)\epsilon_{th}$, which is informative only when $\epsilon_{th}<1/(c-1)$.
\end{theorem}
\begin{proof}
Take $\lambda = 1/2$. By Lemma \ref{lem:mart} we can write $\sqrt{R_t} = G_t\cdot e^{-B_t}$ and the martingale $G_t=G^{1/2}_t$ converges a.s., so $G^*=\sup_tG_t<\infty$. Therefore, by Lemma~\ref{lem:info},
$R_t(q)\le(G^*)^2e^{-2\beta_\delta[N_t(q)+N_t(p)]}$, which vanishes if $N_t(q)+N_t(p)\to\infty$. (i) If $N_t(p)\to\infty$, or under the
strong form, this holds for every $q$, giving $\phi(p,t)=(1+\sum_qR_t(q))^{-1}\to1$. (ii) Suppose $\phi_\infty(q)>\epsilon_{th}$ (the limit exists by Theorem~\ref{thm:safe}), then for all sufficiently large $t$
$\phi(q,t)\ge\epsilon_{th}$. By Assumption \ref{as:AGI},  $p$ or $q$ is sensed in every window of finite length $T$, so $N_t(q)+N_t(p)\to\infty$ which we showed gives $\phi_\infty(q)=0$. Contradiction. Therefore for any $q\neq p$, $\lim_t\phi(q,t)\le\epsilon_{th}$ which implies $\lim_t\phi(p,t)\ge1-(c-1)\epsilon_{th}$.
\end{proof}

\begin{theorem}[Exponential rejection]\label{thm:rate}
Under (A1)--(A3), for every $q\neq p$, a.s. $\liminf_{t\to\infty} -\frac{1}{t}\log\phi(q,t) \geq \liminf_{t\to\infty}\frac{\kappa_0N_t(q)+\kappa_1N_t(p)}{t}$. If a stronger assumption is used, that $p$ or $q$ is sensed at least once in every window of length $T$, the right-hand side becomes at least $\min(\kappa_0, \kappa_1)/T$.
\end{theorem}
\begin{proof}
Again, for $\lambda\in(0,1)$, $G^{\lambda*}=\sup_tG^\lambda_t<\infty$ a.s. By Lemma \ref{lem:mart} $R_t^\lambda = G_t^\lambda e^{-B_t^\lambda}\leq G^{\lambda*}e^{-B_t^\lambda}$. Taking logs, $\lambda \log R_t(q) \leq \log G^{\lambda*} - B_t^\lambda$. By Lemma \ref{lem:info}, we substitute $B_t^\lambda$ and divide by $-\lambda t$, getting: $-\frac{1}{t} \log R_t(q) \geq -\frac{\log G^{\lambda*}}{\lambda t} + \frac{\beta_0^\lambda}{\lambda} \frac{N_t(q)}{t} + \frac{\beta_1^\lambda}{\lambda} \frac{N_t(p)}{t}$.
Define $I_t=\frac1t(\kappa_0N_t(q)+\kappa_1N_t(p))$ and $\rho(\lambda)=\min_i\frac{\beta^\lambda_i}{\lambda\kappa_i}$. By Lemma \ref{lem:info} $\frac{\beta_i^\lambda}{\lambda} \leq \kappa_i \implies\rho(\lambda)\in(0,1]$ and $\frac{\beta_i^\lambda}{\lambda} \geq \rho(\lambda)\kappa_i$, so $-\frac1t\log R_t(q)\geq\rho(\lambda)I_t-\frac{\log G^{\lambda*}}{\lambda t}$. As $t\to \infty$, $\liminf_{t\to\infty}\ -\frac{1}{t}\log R_t(q)  \ \geq \ \liminf_{t\to\infty}\rho(\lambda)I_t$.
This holds simultaneously $\forall\lambda\in\{1/k\}_{k\ge2}$ , where letting $k\to \infty$, $\rho\to1$ by Lemma \ref{lem:info} and use
$\phi(q,t)\le R_t(q)$ to transfer the bound to $\phi(q,t)$. For the stronger case, the assumption implies $N_t(q) + N_t(p) \geq \lfloor t/T\rfloor$ and $I_t\geq \min(\kappa_0, \kappa_1)\frac{(N_t(q) + N_t(p))}{t}\geq \min(\kappa_0, \kappa_1)\frac{\lfloor t/T\rfloor}{t}$. As $t\to \infty$ we get our min rate. 
\end{proof}

This bound is attained at the worst case (gated update within $X_{h,t,active}$, every look at the rim $P_e=\delta$, and $p$ never sensed). Unlike the min-rule rate $\max_vK_v$ of \cite{mitra2020new}, it accumulates over looks since our team pools reports into one belief instead of relying on the best-informed agent. On average, each sensed look removes $\kappa_0\approx\log(1/\delta)$ nats from an empty cell's odds, or adds $\kappa_1\approx\log(1/P_{fa})$ for the target.

\begin{corollary}[Finite-time bound; time and error of confirmation]\label{cor:tau}
Under (A1)--(A3), (i) for all $t$, $m\ge0$, and $r>0$:
$\mathbb{P}(\phi(q,t)\ge r)\le\sqrt{\phi(q,0)/(r\phi(p,0))}\,e^{-\beta_\delta m}+\mathbb{P}(N_t(q)+N_t(p)<m)$. (ii) Let
$\tau_\eta=\inf\{t:\max_q\phi(q,t)\ge1-\eta\}$ and $\hat p$ the confirmed target cell. If every window of length $T$ senses $p$ or all of $X_h \setminus \{p\}$ on every sample path, then with
$\Gamma_0=\sum_{q\neq p}\sqrt{\phi(q,0)/\phi(p,0)}$,
\[
\begin{split}
  \mathbb{E}[\tau_\eta]\le T\big(\lceil\log^+(\Gamma_0/\sqrt\eta)/\beta_\delta\rceil+(1-e^{-\beta_\delta})^{-1}\big),\\
  \mathbb{P}(\hat p\neq p)\le\tfrac{\eta}{1-\eta}\cdot\tfrac{1-\phi(p,0)}{\phi(p,0)} .
\end{split}
\]
(iii) If  $p\sim\phi(\cdot,0)$, then $\mathbb{P}(\tau_\eta<\infty,\ \hat p\neq p)\le\eta$.
\end{corollary}
\begin{proof}
From Lemma \ref{lem:mart}, $G_t = G_t^{1/2} =\sqrt{R_t(q)}e^{B_t}$ is a martingale, so $\mathbb{E}[G_t] = G_0 = \sqrt{R_0(q)}$ since $B_0 = 0$. At $\lambda = 1/2$, from Lemma \ref{lem:info} $B_t \geq \beta_\delta[N_t(q) + N_t(p)]$, with $\beta_\delta = 2.29$. We will also use Markov's inequality. \\
\emph{(i)} Let event $A=\{N_t(q)+N_t(p)\ge m\}$, we have $\mathbb{P}(R_t \geq r) = \mathbb{P}(R_t \geq r, A) + \mathbb{P}(R_t \geq r, A^c)\leq \mathbb{P}(R_t\geq r, A) + \mathbb{P}(A^c)$. Lemma \ref{lem:info} gives $B_t\geq \beta_\delta m$ on $A$, implying $\sqrt{R_t} = G_te^{-B_t}\leq G_te^{-\beta_\delta m}$. $\mathbf{1}_A$ is $1$ on $A$ and $0$ off it so $\sqrt{R_t}\mathbf{1}_A\leq G_t e^{-\beta_\delta m}$ also always holds. So, $\mathbb{E}[\sqrt{R_t}\mathbf{1}_A]\le\sqrt{R_0}e^{-\beta_\delta m}$ by substituting for $\mathbb{E}[G_t]$. The events $\{R_t\geq r, A\} = \{\sqrt{R_t}\mathbf{1}_A\geq \sqrt{r}\}$ and Markov's Inequality give $\mathbb{P}(\sqrt{R_t}\mathbf{1}_A\geq \sqrt{r}) \leq \mathbb{E}[\sqrt{R_t}\mathbf{1}_A]/\sqrt{r}$. Putting these together: $\mathbb{P}(R_t\geq r, A)\leq \frac{\mathbb{E}[\sqrt{R_t}\mathbf{1}_A]}{\sqrt{r}}\leq \sqrt{\frac{R_0}{r}}e^{-\beta_\delta m}$. After adding $\mathbb{P}(A^c)$ and using $\phi(q,t)\leq R_t(q)$ we realize (i). \\
\emph{(ii)} Take $m = \lfloor t/T \rfloor$ after $t$ steps since every window of $T$ steps contains a sensed look at $q$ or $p$ for every $q\neq p$, so $N_t(q) + N_t(p) \geq \lfloor t/T \rfloor$. Then, $A$ always occurs and by steps in (i), $\mathbb{E}[\sqrt{R_t(q)}]\leq \sqrt{R_0(q)}e^{-\beta_\delta \lfloor t/T\rfloor}$ for every wrong cell $q$. If by step $t$, the team has not confirmed, then $\phi(p,t) < 1 - \eta$ and $1-\phi(p,t) = \frac{\sum_q R_t(q)}{1 + \sum_q{R_t(q)}}\leq \sum_q R_t(q)$, which means $\mathbb{P}(\tau_\eta>t)\leq \mathbb{P}(\sum_{q\neq p}R_t(q) \geq \eta)$  because $\{\tau_\eta > t\}\subseteq\{\phi(p,t)< 1 - \eta\}\subseteq\{\sum_{q\neq p}R_t(q) \geq \eta\}$. By Markov's inequality, subadditivity of square roots, and by our bound of the expectation above, $\mathbb{P}(\sum_q R_t(q)\geq \eta) = \mathbb{P}(\sqrt{\sum_q R_t(q)}\geq \sqrt{\eta})\leq \frac{\sum_q\mathbb{E}\sqrt{R_t(q)}}{\sqrt{\eta}}\leq \frac{\Gamma_0}{\sqrt{\eta}}e^{-\beta_\delta \lfloor t/T\rfloor}$ where $\Gamma_0 = \sum_{q\neq p} \sqrt{R_0(q)}$ ($\Gamma_0$ can be seen as a measure of how spread out the prior is). Let $a = \Gamma_0/\sqrt{\eta}$ , $\tau_\eta$ is a nonnegative integer random variable, so $\mathbb{E}[\tau_\eta] = \sum_{t\geq0}\mathbb{P}(\tau_\eta > t)\leq\sum_{t\geq0}\min\{1, ae^{-\beta_\delta \lfloor t/T\rfloor}\} = T\sum_{k\geq 0}\min\{1, ae^{-\beta_\delta k}\}$ where $k = \lfloor t/T\rfloor$ is exactly the $k$-th window of $T$ timesteps. We re-index the sum over windows. Let $k_0 = \lceil \log^+ a/\beta_\delta \rceil$, such that $k_0$ is the first window where $ae^{-\beta_\delta k} \leq 1$. For $k\geq k_0$, $ae^{-\beta_\delta k} = (ae^{-\beta_\delta k_0})e^{-\beta_\delta(k - k_0)}\leq e^{-\beta_\delta(k - k_0)}$. We can write $\sum_{k\geq 0} \min\{1, ae^{-\beta_\delta k}\}\leq k_0 + \sum_{j\geq 0} e^{-\beta_\delta j} = k_0 + \frac{1}{1-e^{-\beta_\delta}}$. Combining these two inequalities we get the first bound on $\mathbb{E}[\tau_\eta]$ in (ii). The probability the confirmed cell is wrong is at most the probability a wrong cell reaches belief $\geq 1 - \eta$ at some time, which is bounded by Theorem \ref{thm:safe}.\\
\emph{(iii)} Since target is drawn from $\phi(\cdot, 0)$ (so that $\mathbb{P}$ also averages over the target's location), $\phi(\cdot, t)$ is the true posterior of target's location given $\mathcal{F}_t$ for any $t$, including $\tau_\eta$ (i.e. $\mathbb{P}(\text{target at } q \ | \ \mathcal{F}_t) = \phi(q,t)$). On $\{\tau_\eta < \infty\}$, $\phi(\hat{p},\tau_\eta) \geq 1- \eta$, so $\mathbb{P}(\hat{p} \neq p \ | \ \mathcal{F}_{\tau_\eta}) = 1 - \phi(\hat{p},\tau_\eta)\leq \eta$. By law of total probability, averaging probability over all confirming runs gives: $
  \mathbb{P}(\tau_\eta<\infty,\ \hat p\neq p)
  =\mathbb{E}\Big[\mathbf{1}\{\tau_\eta<\infty\}\,\underbrace{\mathbb{P}(\hat p\neq p\,|\,\mathcal{F}_{\tau_\eta})}_{\le\,\eta}\Big]
  \ \le\ \eta\,\mathbb{E}\big[\mathbf{1}\{\tau_\eta<\infty\}\big]
  =\eta\,\mathbb{P}(\tau_\eta<\infty)\ \le\ \eta .$
\end{proof}

\begin{remark}[Scope and limitations]\label{rem:scope}

\emph{Communication.} We assume lossless, instantaneous sharing within $r_c$. \emph{Splits/Breaks in Connectivity.} (A2) fails while the team is split, but each sub-team holds the exact posterior of its own reports, so results hold per sub-team with its own look counts.
\emph{Misspecification.} Results assume $P_d$ and $P_{fa}$ are modeled exactly; otherwise $R_t(q)$ is no longer a
martingale and Theorem~\ref{thm:safe} can fail. \emph{Target.} Results assume a single, static target and independent reports over time.
\end{remark}

\section{SIMULATIONS}
\label{sec:sim}
\begin{table}[!t]
\renewcommand{\arraystretch}{1.1}
\setlength{\tabcolsep}{2pt}
\caption{Simulation parameters}
\label{tab:params}
\centering
\footnotesize
\begin{tabular}{llll}
\hline
Map $X$ & $[0,5]^2$, $200\times200$ cells & Agents $n$ & 10\\
$P_{fa}$, $\delta$ & $10^{-6}$, 0.01 & $P_e=\delta$ at & 0.15 (one agent)\\
$r_c$, $d_c$ & 0.1, 0.05 & $r$ in $L$, $s$ & 0.08, 20\\
$d_{\min}$, $d_{\max}$ & 0.005, 0.006 & $\sigma$, $\epsilon$ & 0.1, $10^{-6}$\\
$\rho$, $w$ & 1.67, $10^{-4}$ & $\tau$, $\sigma_k$ & 1, 0.1\\
Prior floor & 0.1 & $\eta$ & 0.01\\
\hline
\end{tabular}
\end{table}
\subsection{Simulation Setup}
All experiments were run in MATLAB, using the same world scale, sensor model, and belief update. The prior $\phi(\cdot, 0)$ is a mixture of anisotropic Gaussian hot spots of different sizes and orientations plus a $10\%$ uniform floor, ensuring every cell has positive prior mass. The $n = 10$ agents start in a connected formation near the bottom left corner. We adopt the sensing cost $f_{local}(\|q-x_i(t)\|)=\tfrac12\|q-x_i(t)\|^2+\epsilon$, standard in coverage control \cite{schwager2006distributed,schwager2009decentralized}, where $\epsilon>0$ enforces $f_{local}>0$. We assign $\epsilon = 10^{-6}$ and the pull region uses $r_s = 0.15$, $\gamma = 1.1$, and $\varepsilon_g = 10^{-12}$. Agents communicate within $r_c = 0.1$ (inside the Laplacian we use $0.08$ as a safety margin), with $\zeta = 1$. A cell is confirmed when its belief reaches $1-\eta$, with $\eta = 0.01$. The remaining parameters can be found in Table \ref{tab:params}.

We run five experiments. \emph{(i) Single run:} a fixed target, to show the belief evolution, trajectories, and connectivity (Fig. \ref{fig:search_results}). \emph{(ii) AGI Test:} we measure the probability a cell with belief $>\varepsilon_{th}$ waits more than $s$ timesteps throughout multiple trials using PPALA to test how well our AGI assumption holds (Fig.~\ref{fig:agi}). \emph{(iii) Lawnmower comparison: }$30$ paired trials in which PPALA and a connected lawnmower sweep search for the same target, drawn from a prior concentrated around the top third of the map, with identical sensing, false alarms, and belief update. The lawnmower agents move in a line (spacing $0.08$) at full speed $d_{max}$ along boustrophedon lanes that sense every cell in each sweep, ignoring the belief. \emph{(iv) Theory check:} we compare measured evidence per look, rejection rates, and wrong-confirmation frequencies with Theorems \ref{thm:conv}--\ref{thm:rate} and Corollary \ref{cor:tau}, over $100$ PPALA trials at $P_{fa} = 10^{-6}$ and $400$ at $P_{fa} = 10^{-2}$, plus $50$ lawnmower trials at each rate. At $P_{fa} = 10^{-2}$, $\gamma_s$ is rescaled so the sense radius stays the same (Fig.~\ref{fig:convergence}  and  ~\ref{fig:wrong}). \emph{(v) Ablation}: we remove individual components of our algorithm and compare the changing results in target acquisition under different priors  (Fig.~\ref{fig:ablation_priors}, Table~\ref{tab:ablation}). We consider the following variants: ``w/o connectivity" (we remove connectivity constraints), ``w/o adaptive $R_G$" (we keep the pull region constant), ``$R_G = X_h$" (we set the pull region to the entire space), ``greedy search" (we replace $\mathcal{H}_t(x)$ by a greedy solution that chooses to go to individual cells with the greatest belief). Trials in each experiment end at confirmation or at a step limit ($20,000$ steps in (ii, iv), $15,000$ steps in (iii, v)).  
\begin{figure}[!t]
\centering
\includegraphics[width=0.32\columnwidth]{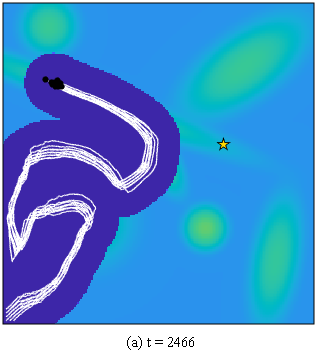}\hfill
\includegraphics[width=0.32\columnwidth]{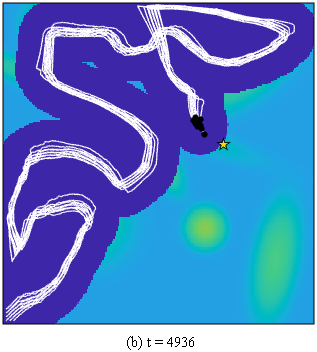}\hfill
\includegraphics[width=0.32\columnwidth]{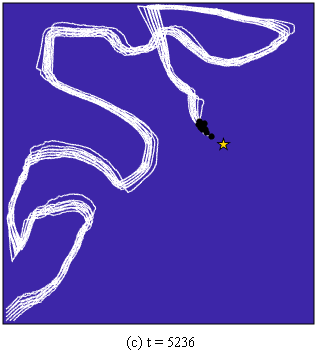}
\caption{Belief $\phi$ (color) and robot trajectories (white) at three times of a single run; the star marks the target.}
\label{fig:search_results}
\end{figure}
\begin{figure}[!t]
\centering
\includegraphics[width=0.8\columnwidth]{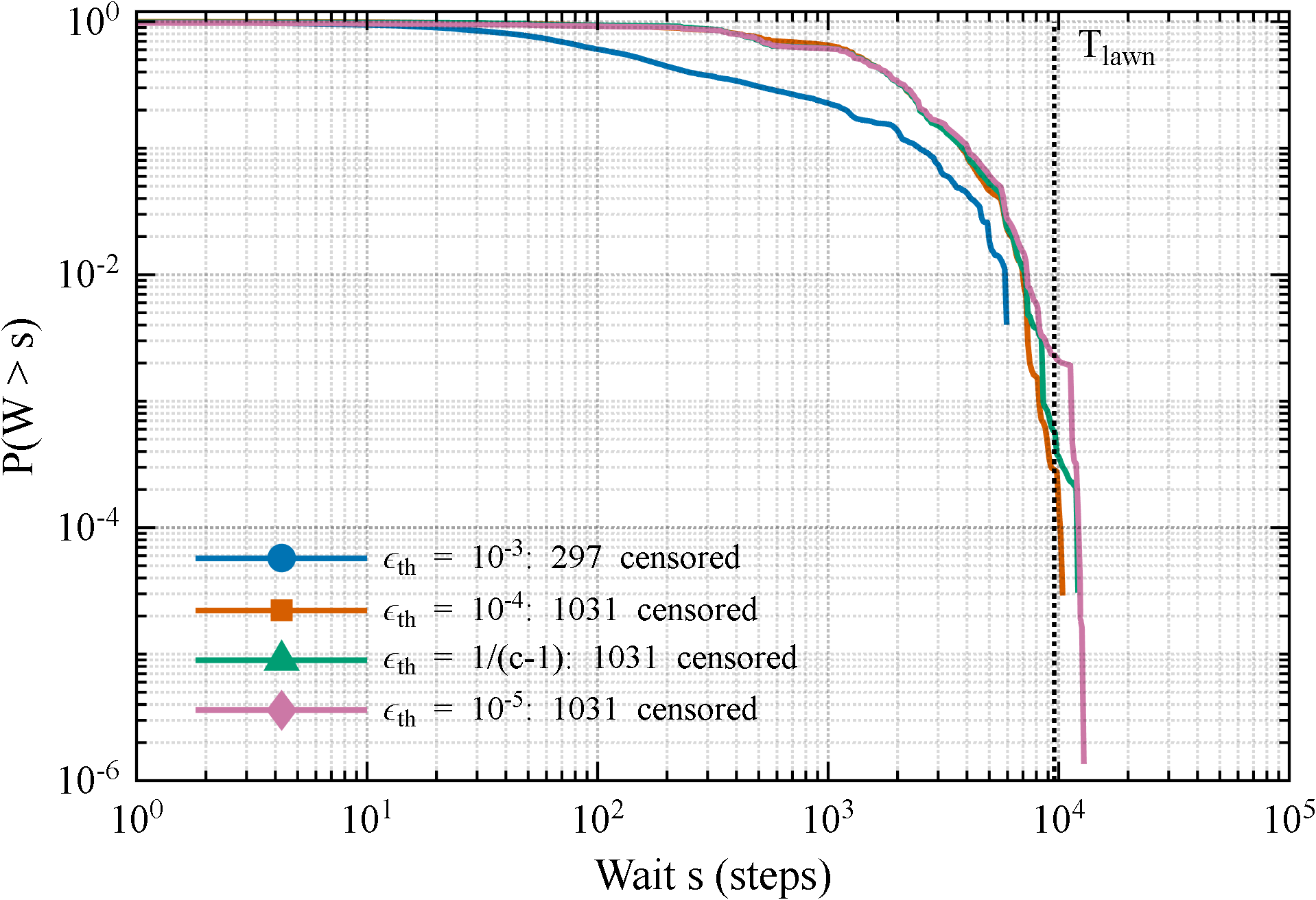}
\caption{Empirical check of Assumption~\ref{as:AGI} for PPALA ($40$ trials, target drawn from prior). An episode starts when a cell $q\neq p$ first has $\phi(q,t)\ge\epsilon_{th}$ and ends when $q$ or $p$ is sensed. $W$ is its wait. Open episodes when a trial ends are censored. Dotted: lawnmower's coverage window $T_{lawn}$, which satisfies the assumption for every $\epsilon_{th}$.}
\label{fig:agi}
\end{figure}
 \begin{figure}[!t]
\centering
\includegraphics[width=\columnwidth]{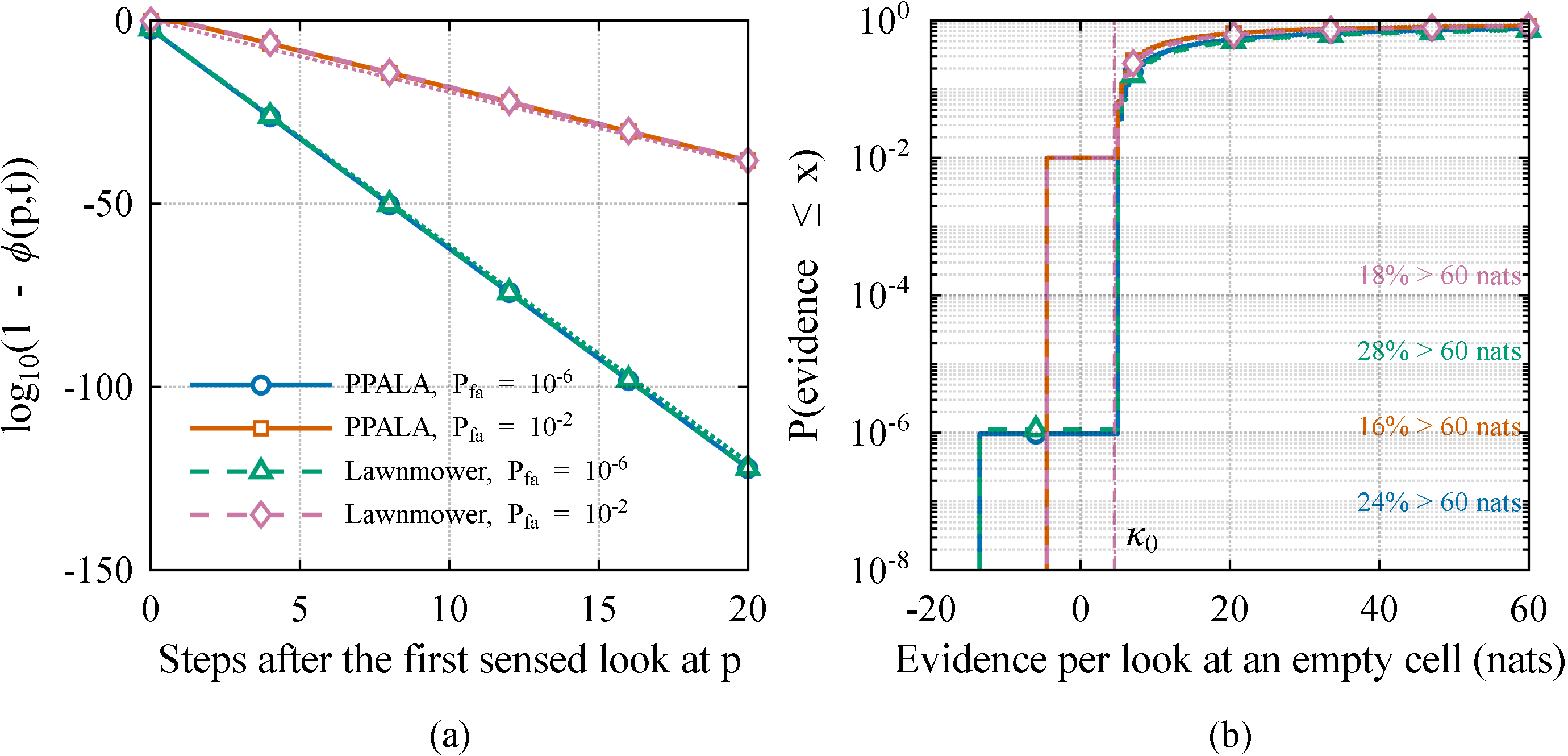}
\caption{(a) Residual belief after first sensed look at target (multi-trial median); dotted: slope $-\kappa_1/\ln 10$ per look. (b) Empirical CDF of the evidence of one look at an empty cell. Vertical line: $\kappa_0$. Solid: PPALA; dashed: lawnmower.}
\label{fig:convergence}
\end{figure}
\begin{figure}[!t]
\centering
\includegraphics[width=0.8\columnwidth]{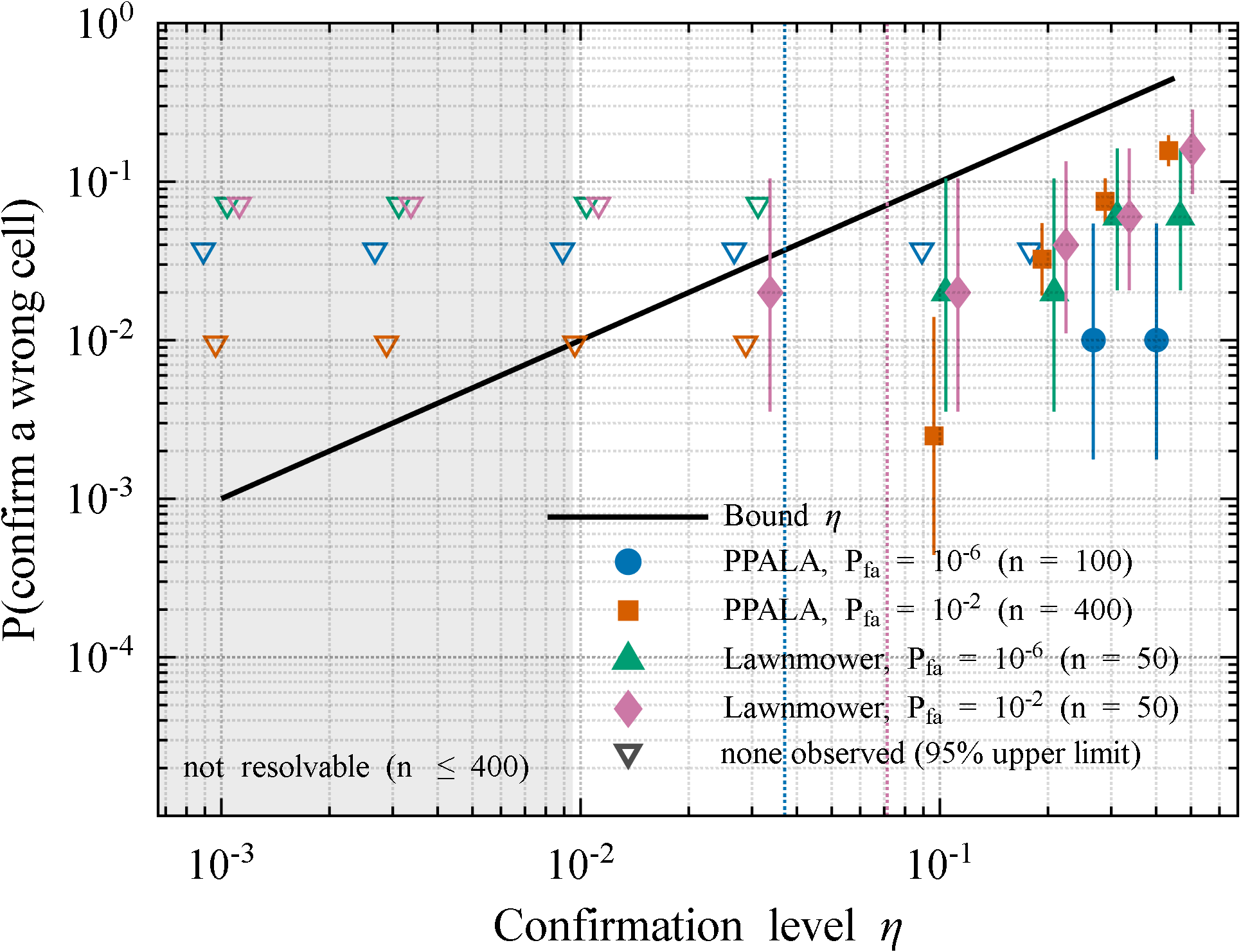}
\caption{Frequency of confirming a wrong cell vs.\ confirmation level, showing bound $\eta$ of Corollary~\ref{cor:tau}(iii) (target drawn from the prior). Filled markers: observed frequency with 95\% Wilson intervals; open triangles: no wrong confirmation observed (95\% upper limit). Shaded region and dotted lines: not enough trials to resolve the bound.}
\label{fig:wrong}
\end{figure}
\begin{table}[!t]
\centering
\includegraphics[width=\columnwidth]{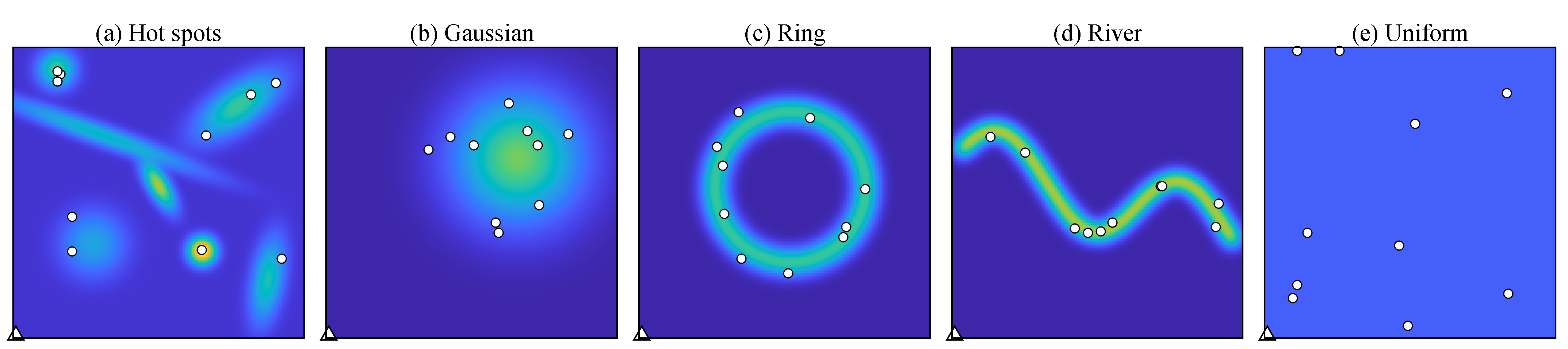}
\captionof{figure}{Priors of the ablation (brighter = more likely), with the targets of their 10 trials (circles) and the team's start location (triangle).}
\label{fig:ablation_priors}
\vspace{4pt}
\renewcommand{\arraystretch}{1.1}
\setlength{\tabcolsep}{2pt}
\captionof{table}{Ablation of PPALA (10 trials per prior)}
\label{tab:ablation}
\centering
\small
\begin{tabular}{@{}lccccc@{}}
\hline
Variant & Hot spots & Gaussian & Ring & River & Uniform\\
\hline
Full PPALA & $\mathbf{3.9}$ & $2.5$ & $\mathbf{1.4}$ & $\mathbf{1.4}$ & $4.2$\\
\hline
w/o connectivity & $7.6$ & $4.0$ & $2.1$ & $2.0$ & $11.0$\\
w/o adaptive $R_G$ & $9.3$\,(6) & $\mathbf{2.4}$ & $6.0^{*}$ & $13.8^{*}$\,(5) & $9.8$\,(6)\\
$R_G=X_h$ & $18.0^{*}$\,(1) & $16.2^{*}$\,(2) & $18.0^{*}$\,(1) & $16.1$\,(2) & $16.0$\,(2)\\
greedy search & $10.0$\,(7) & $10.9^{*}$\,(8) & $5.9$\,(8) & $2.6$ & $\mathbf{3.6}$\\
\hline
\end{tabular}

\vspace{2pt}
\parbox{\columnwidth}{\scriptsize Mean steps to find (thousands) until all agents found the target. (k): found in only k of 10 trials (not found counts as 20). \textbf{Bold:} fastest for the prior. $^{*}$: Significantly differs from full PPALA, $p<0.05$ (paired Wilcoxon test, Holm correction per prior).}
\end{table}
\subsection{Results and Discussion}
\emph{(i)} Fig. \ref{fig:search_results} shows a single run. Starting connected from the corner, the team sweeps the map as one network, and the exact Bayes update clears the belief along its sensed path (dark region). The belief of unvisited cells, including the target's, grows slowly as the cleared mass is renormalized over them. Once the team senses the target, within a few looks, the belief there rises to $1$ and vanishes everywhere else. 

\emph{(ii)} Fig. \ref{fig:agi} tests whether PPALA satisfies Assumption \ref{as:AGI}, the condition under which Theorem \ref{thm:conv}(ii) holds. Higher belief cells are visited sooner. The probability of seeing waits greater than the lawnmower's deterministic window drops off very abruptly. PPALA does not enforce revisits, so Assumption \ref{as:AGI} is a condition on the planner that our runs satisfy in most trials. This is consistent with treating the planner as a black box in Section~\ref{sec:conv}.

\emph{(iii)} Next we compared PPALA with a connected lawnmower sweep over $30$ paired trials. Both planners confirmed the correct cell in every trial at $\eta = 0.01$. PPALA was faster in all $30$, with  a median of $1659$ steps with interquartile range (IQR) $1327$--$1888$ versus $3521$ with IQR $3356$--$3881$ for the lawnmower strategy. The advantage comes from PPALA being belief-aware, moving to high-belief areas first while the lawnmower sweeps in a fixed order.

\emph{(iv)} Fig. \ref{fig:convergence} checks the rates of Section V. Once the target is sensed, $\phi(p,t)\to 1$ (Theorem \ref{thm:conv}(i)) for both planners, and $ 1 - \phi(p,t)$ decays at the rate set by the $\kappa_1N_t(p)$ term of Theorem \ref{thm:rate} (at about $5.9$ decades per look at $P_{fa} = 10^{-6}$ and $2.0$ at $P_{fa} = 10^{-2}$, Fig. \ref{fig:convergence}a). Fig. \ref{fig:convergence}b illustrates Lemma \ref{lem:info}. Every correctly sensed empty cell yields at least $\log((1-P_{fa})/\delta)\approx \kappa_0$ nats (only looks below this are false alarms, occurring at modeled rates for $P_{fa}$).
Fig. \ref{fig:wrong} tests Corollary \ref{cor:tau}(iii). In every condition and at every $\eta$, the observed frequency of confirming a wrong cell lies below the bound. At the operating level $\eta = 0.01$, PPALA made no wrong confirmations in $400$ trials even at $P_{fa} = 10^{-2}$. Wrong confirmations appear only for $\eta\geq 0.03$, and for PPALA at $P_{fa} = 10^{-6}$ only for $\eta \geq 0.3$.

\emph{(v)} Fig.~\ref{fig:ablation_priors} shows the five priors used for the ablation. In Table~\ref{tab:ablation} PPALA was not the absolute fastest for only the single Gaussian and the uniform priors. Even so, the difference from the fastest variant was not statistically significant, and full PPALA found the target in all ablation trials.

\section{CONCLUSION}
We presented a coordination framework for multi-agent search that fuses sensor uncertainties by inverse-variance weighting and keeps the team connected through a differentiable constraint maintained by a primal-dual planner. For any planner that does not anticipate future measurements, we bounded the probability of ever confirming a wrong location and showed that false locations are rejected exponentially fast in the number of looks. Under AGI, the belief provably concentrates on the target. In simulation, our planner maintained connectivity, satisfied AGI in most trials, and located the target faster than a lawnmower sweep when the prior was informative. The measured residual belief, evidence per look, and false-confirmation rates all agreed with our bounds.

{\footnotesize
\bibliographystyle{IEEEtran}       %
\bibliography{arXiv/ref}}

\end{document}